\documentclass[11pt]{article}
\usepackage{fullpage}
\usepackage{amsfonts, amsthm, amssymb, amsmath, algorithmic, algorithm}
\usepackage{helvet}

\usepackage{amsmath,amssymb}
\usepackage{amsfonts}
\newcommand {\ignore} [1] {}

\usepackage{epsfig}
\usepackage{wrapfig}

\usepackage[pdftex, pagebackref=false, colorlinks=true, urlcolor=black, citecolor=black, linkcolor=black]{hyperref}

\usepackage{framed}

\usepackage{verbatim}
\usepackage{fullpage}

\usepackage{epsfig}

\usepackage{verbatim}

\newtheorem{thm}{Theorem}[section]
\newtheorem{lem}[thm]{Lemma}
\newtheorem{cor}[thm]{Corollary}
\newtheorem{defn}[thm]{Definition}
\newtheorem{clm}[thm]{Claim}
\newtheorem{cons}[thm]{Construction}
\newtheorem{prop}[thm]{Yroposition}

\newenvironment{theorem}{\begin{thm}\begin{rm}}%
{\end{rm}\end{thm}}
\newenvironment{lemma}{\begin{lem}\begin{rm}}%
{\end{rm}\end{lem}}
\newenvironment{corollary}{\begin{cor}\begin{rm}}%
{\end{rm}\end{cor}}
{\end{em}\end{defn}}
\newenvironment{claim}{\begin{clm}\begin{rm}}%
{\end{rm}\end{clm}}
{\end{em}\end{cons}}
{\end{rm}\end{prop}}

\newcommand{\secref}[1]{\hyperref[#1]{Section \ref{#1}}}
\newcommand{\apref}[1]{\hyperref[#1]{Appendix \ref{#1}}}
\newcommand{\thref}[1]{\hyperref[#1]{Theorem \ref{#1}}}
\newcommand{\defref}[1]{\hyperref[#1]{Definition \ref{#1}}}
\newcommand{\lemref}[1]{\hyperref[#1]{Lemma \ref{#1}}}
\newcommand{\clref}[1]{\hyperref[#1]{Claim \ref{#1}}}
\newcommand{\consref}[1]{\hyperref[#1]{Construction \ref{#1}}}
\newcommand{\figref}[1]{\hyperref[#1]{Figure \ref{#1}}}
\newcommand{\eqnref}[1]{\hyperref[#1]{Equation \ref{#1}}}

\def\poly{\operatorname{poly}}
\def\opt{{\sf OPT}}

\newcommand{\eat}[1]{}
\newcommand{\floor}[1]{\left\lfloor #1 \right\rfloor}
\newcommand{\ceil}[1]{\left\lceil #1 \right\rceil}

\newcommand{\set}[1]{\{ #1 \}}
\newcommand{\pr}[2]{{\bf Pr}_{#1}[ #2 ]}
\newcommand{\abs}[1]{|#1|}

\newcommand{\squishlist}{\begin{itemize}}
\newcommand{\squishend}{\end{itemize}}

\def\cV{\mathcal{V}}

\def\danupon#1{}
\def\parinya#1{}

\title{Graph Pricing Problem on Bounded Treewidth, Bounded Genus and $k$-Partite Graphs}
\author{
Parinya Chalermsook\footnote{Max-Planck-Institut f\"{u}r Informatik, Saarbr\"{u}cken, Germany. Email: {\tt parinya@cs.uchicago.edu}. Work partially done while at the Department of Computer Science, University of Chicago, Chicago, IL. }\ \ \
Shiva Kintali\footnote{Department of Computer Science, Princeton University, Princeton, NJ-08540. Email : {\em{kintali@cs.princeton.edu}}}\ \ \
Richard Lipton\footnote{College of Computing, Georgia Institute of Technology, Atlanta, GA-30332. Email : {\em{rjl@cc.gatech.edu}}}\ \ \
Danupon Nanongkai\footnote{School of Physical and Mathematical Sciences, Nanyang Technological University, Singapore. Email : {\em{danupon@gmail.com}}. Supported in part by the following research grants: Nanyang Technological University grant M58110000, Singapore Ministry of Education (MOE) Academic Research Fund (AcRF) Tier 2 grant MOE2010-T2-2-082, and Singapore MOE  AcRF Tier 1 grant MOE2012-T1-001-094. Work partially done while at Georgia Institute of Technology, USA, and University of Vienna, Austria.}\\
}
\date{}

\begin{document}

\maketitle

%\vspace{0.2in}
%\begin{center}
%{\large{\bf{Abstract}}}
%\end{center}

\begin{abstract}
Consider the following problem. A seller has infinite copies of $n$ products represented by nodes in a graph. There are $m$ consumers, each has a budget and wants to buy two products. Consumers are represented by weighted edges. Given the prices of products, each consumer will buy both products she wants, at the given price, if she can afford to. Our objective is to help the seller price the products to maximize her profit.

This problem is called {\em graph vertex pricing} ({\sf GVP}) problem and has resisted several recent attempts despite its current simple solution. This motivates the study of this problem on special classes of graphs. In this paper, we study this problem on a large class of graphs such as graphs with bounded treewidth, bounded genus and $k$-partite graphs.

We show that there exists an {\sf FPTAS} for {\sf GVP} on graphs with bounded treewidth. This result is also extended to an {\sf FPTAS} for the more general {\em single-minded pricing} problem. On bounded genus graphs we present a {\sf PTAS} and show that {\sf GVP} is {\sf NP}-hard even on planar graphs.

%After a certain preprocessing step, this problem can be reduced to a 2-CSP over large integral domain.   

%We study the integrality gap of the Sherali-Adams relaxation of the natural {\sf LP}  applied after preprocessing. 

%We study the Sherali-Adams relaxation of a natural {\sf LP} that $(1+\epsilon)$-approximates {\sf GVP}. 

We study the Sherali-Adams hierarchy applied to a natural Integer Program formulation that $(1+\epsilon)$-approximates the optimal solution of {\sf GVP}. 
Sherali-Adams hierarchy has gained much interest recently as a possible approach to develop new approximation algorithms. We show that, when the input graph has bounded treewidth or bounded genus, applying a constant number of rounds of Sherali-Adams hierarchy makes the integrality gap of this natural {\sf LP} arbitrarily small, thus giving a $(1+\epsilon)$-approximate solution to the original {\sf GVP} instance.

%makes the integrality gap of the natural {\sf LP} 

%arbitrarily close to one.

On $k$-partite graphs, we present a constant-factor approximation algorithm. We further improve the approximation factors for paths, cycles and graphs with degree at most three.
\end{abstract}

\thispagestyle{empty}
\newpage
\setcounter{page}{1}

\section{Introduction}

Consider the following problem where a seller is trying to sell her products to make the most profit. The seller has infinite copies of $n$ products represented by nodes in a graph $G(V,E)$. She knows that there are $m$ consumers who want to buy exactly two products each. Each consumer is represented by an edge $e$ in $G$ and has a budget $B_e$. She will buy both products (represented by the end vertices of the edge $e$) if the price of both products together does not exceed her budget $B_e$. The seller's goal is to price all products to make the most revenue. That is, she wants to find a price function $p:V(G)\rightarrow \mathbb{R}_+\cup \{0\}$ that maximizes
\[\sum_{uv\in E(G)} \begin{cases}p(u)+p(v) &\mbox{if $p(u)+p(v)\leq B_e$,}\\ 0 &\mbox{otherwise.}
\end{cases}\]

This problem is called the {\em graph vertex pricing} ({\sf GVP}) problem. It is one of the fundamental special cases of the {\em single-minded item pricing} ({\sf SMP}) problem. In {\sf SMP}, consumers may want to buy more than two products (thus we can represent the input by a hypergraph with budgets on the hyperedges). This problem arises from the application of pricing digital goods (see \cite{GuruswamiHKKKM05} for more details).

Both {\sf GVP} and {\sf SMP} are proposed by Guruswami et al. \cite{GuruswamiHKKKM05} along with an $O(\log n+\log m)$ approximation algorithm for the {\sf SMP} problem and an {\sf APX}-hardness of the {\sf GVP} problem. 
Balcan and Blum \cite{BalcanB07} presented a surprisingly simple algorithm achieving a factor four approximation for the {\sf GVP} problem. 
The hardness of approximation of {\sf GVP} was recently shown to be $2-\epsilon$, assuming the Unique Games Conjecture \cite{KhandekarKMS09}. Even more recently, the hardness of $k^{1-\epsilon}$ for {\sf SMP} was shown by \cite{ChalermsookLN13FOCS} (building on \cite{BriestK11,ChalermsookLN13soda,ChalermsookCKK12}), assuming ${\sf P}\neq {\sf NP}$ where $k$ is the size of the largest hyperedge.

%\danupon{I'm still not sure whether we should have this paragraph. But simple algorithm should makes the reader see why it's tempting to improve the result.}
The algorithm of Balcan and Blum first constructs a bipartite graph by randomly partitioning the vertices into two sides and deleting edges connecting vertices on the same side. It then picks one side randomly and prices all vertices on that side to zero. The resulting instance can be solved optimally. The approximation guarantee of four follows from the fact that each edge is deleted with probability half and pricing vertices on one random side to zero reduces the optimal revenue by half (in expectation). The algorithm can be derandomized using standard techniques.

Understanding special cases might lead to improving the upper bound for the general case, as understanding the case of bipartite graph lead to a $4$-approximation algorithm for the general case. In general, it is interesting to explore how the combinatorial structure of the input graph influences the approximability of the problem. This line of attack was initiated recently in \cite{KhandekarKMS09} where it is shown that the {\sf GVP} problem is {\sf APX}-hard on bipartite graphs, and in \cite{KrauthgamerMR11} where an improved algorithm is presented for the case where the range of consumers budgets is restricted.

In this paper, we continue this line of study and present approximation algorithms and hardness results of the {\sf GVP} problem on many classes of graphs such as bounded treewidth graphs, bounded genus graphs and $k$-partite graphs.

%While much progress has been made on many special cases of the single-minded item pricing problem when input structures are restricted (see, e.g., CITE: Highway, What else?)
%Write motivation here as in Khandekar et al.
%In this paper, we study three large classes of graphs, i.e., bounded genus, bounded tree-width, and bounded degree graphs.

\subsection{Our Results}

\paragraph{Bounded treewidth graphs and hypergraphs:} To understand the structure of {\sf GVP}, trees are a natural class of graphs to study. {\sf GVP} is not known to be NP-hard on trees. We present an {\sf FPTAS} for trees based on a dynamic programming algorithm. We generalize our algorithm to bounded treewidth graphs by applying our technique on the input graph's tree decomposition, which can be computed in linear time by an algorithm of Bodlaender~\cite{Bodlaender96}. We extend our algorithm to give an {\sf FPTAS} for solving {\sf SMP} on bounded treewidth {\em{hypergraphs}} as well.

\paragraph{Bounded genus graphs:} Bounded genus graphs are broad class of graphs which play a major role both in structural graph theory and algorithmic graph theory. Several important {\sf NP}-hard optimization problems admit improved algorithms on planar graphs (i.e., graphs with genus zero). We present a {\sf PTAS} for {\sf GVP} on bounded genus graphs and show that it is {\sf NP}-hard even on planar graphs.

Our result relies on Baker's and Eppstein's Techniques \cite{Baker94,Eppstein00}. However, while such previous results presented a polynomial time exact algorithms on graphs with bounded treewidth, it does not seem to be the case in the {\sf GVP} problem. Instead, we show that the {\sf GVP} admits fully polynomial-time approximation scheme ({\sf FPTAS}) on bounded treewidth graphs.

To the best of our knowledge, the only other work that studied any variation of the pricing problem on graphs with bounded treewidth or bounded genus is \cite{CardinalDFJNW09}, where it is shown that the {\em Stackelberg minimum spanning tree pricing} ({\sf SMST}) problem can be solved in polynomial time on bounded treewidth graphs and {\sf NP}-hard on planar graphs. The question whether there is a {\sf PTAS}  for {\sf SMST} on planar graphs is still open. Thus, our result is the first {\sf PTAS} for a pricing problem on planar graphs and it extends to bounded genus graphs.

%The main difficulty is that there might be exponentially number prices one could assign to a vertex. However, we show that one can get {\sf FPTAS} for this case.
%Many of these results rely on the fact that many problems can be solved in polynomial time on graphs of bounded tree-width. Thus, the problem on bounded tree-width graphs is essential to study as well.
%
%For the case of bounded tree-width graphs,
%Many of these results rely on Baker's technique \cite{Baker94} and the fact that many problems can be solved in polynomial time on graphs of bounded {\em tree-width}.

\paragraph{Integrality gap of the Sherali-Adams relaxation:} We study the integrality gap of the linear program in the Sherali-Adams hierarchy \cite{SheraliA90}. Sherali-Adams relaxation is one of the lift-and-project schemes which have received much interest recently in the approximation algorithms community (see, e.g \cite{KarlinMN10} and references therein). A question of particular interest is how the integrality gaps evolve through a series of rounds of Sherali-Adams lift-and-project operations. Positive results (i.e., small integrality gap) could potentially lead to an improved algorithm while negative results rule out a wide class of approximation algorithms.

In this paper, we study the Sherali-Adams hierarchy applied to a natural Integer Program formulation that $(1+\epsilon)$-approximates the optimal solution of {\sf GVP} (the $(1+\epsilon)$ approximation factor is needed to be able to write an LP). 
We show a positive result that  the integrality gap after we apply $O_\epsilon(\min(g, w))$ rounds of Sherali-Adams lift and project is one for the case of bounded-treewidth graphs and $(1+\epsilon)$ for the case of bounded-genus graphs, for any $\epsilon$, where $g$ and $w$ are the genus and the treewidth of the input graph respectively and the constant in $O_\epsilon$ depends on $\epsilon$. Thus the Sherali-Adams hierarchy gives us a $(1+\epsilon)$-approximate solution to the original {\sf GVP} instance for these two cases. 
%
%We show a positive result that the integrality gap of this LP is $(1+\epsilon)$, for any $\epsilon$, after $O_\epsilon(\min(g, w))$  rounds of Sherali-Adams hierarchy where $g$ and $w$ are the genus and the treewidth of the input graph respectively and the constant in $O_\epsilon$ depends on $\epsilon$.
%
To our knowledge, besides our work, the only work that studies the Sherali-Adams relaxation in the realm of pricing problem is \cite{KhandekarKMS09} where a lower bound of $4-\epsilon$ on the integrality gap of a natural LP is shown for the general case. It can be observed that this LP from \cite{KhandekarKMS09} is equivalent to the result of two rounds of Sherali-Adams lift-and-project operation on our LP. Our work is the first result that studies the integrality gap of this problem when {\em{many}} rounds of Sherali-Adams hierarchy are allowed and suggests that Sherali-Adams relaxation might be useful in attacking the {\sf GVP} problem on general graphs.

Sherali-Adams relaxation for combinatorial optimization problems on bounded treewidth graphs has been considered before, most notably in~\cite{Treewidth-SheraliAdams,SparsestCut-Treewidth,WainwrightJ04}; in particular, a bounded number of rounds in the Sherali-Adams hierarchy is known to be tight, e.g., for constraint satisfaction satisfaction problems (CSP)~\cite{WainwrightJ04}.
With some pre-processing, {\sf GVP} can be seen as a $2$-CSP when the domain size is large, so the gap upper bound in bounded treewidth graphs follows immediately.   
%Our rounding algorithm (although discovered independently) is similar to the algorithm in~\cite{Treewidth-SheraliAdams}. 
We refer the readers to \cite{Treewidth-SheraliAdams,SparsestCut-Treewidth,WainwrightJ04} for a more complete survey on Sherali-Adams relaxation.

%this could be a potential approach to attack the {\sf GVP} problem on general graphs.

\paragraph{$k$-partite graphs:} We also study the {\sf GVP} problem on $k$-partite graphs. This class of graphs generalizes
the class of bipartite graphs used to develop the 4-approximation algorithm of \cite{BalcanB07} and includes graphs of bounded degree (due to Brooks' theorem~\cite{Brook41}) and graphs of bounded genus (due to Heawood's result \cite{heawood90}).

%
%\danupon{To do: Show that for some approximation ratio of $k$-partite graph, we can improve Balcan-Blum's algorithm. }

Since the {\sf GVP} problem is {\sf APX}-hard even on bipartite graphs \cite{KhandekarKMS09}, we cannot hope for a {\sf PTAS} for $k$-partite graphs. We present a $(4\cdot\frac{k-1}{k})$-approximation algorithm when $k$ is even and a $(4\cdot\frac{k}{k+1})$-approximation when $k$ is odd. This gives a slight improvement on the approximation factor on general graphs. We show that improving these bounds further, when $k$ is even, is as hard as improving the 4-approximation factor for the general graphs.

\paragraph{Bounded Degree graphs:} Finally, we study the problem on graphs of degree at most two and three. For graphs of degree at most two (i.e., paths and cycles), we show that the {\sf GVP} problem can be solved optimally in polynomial time. For graphs of degree at most three, the previous result on tri-partite graphs implies a $3$-approximation algorithm since these graphs are 3-colorable (by Brook's theorem \cite{Brook41}). We present a different algorithm that is a 2-approximation for this class of graphs.

\paragraph{Unit-Demand Min-Buying Pricing} We remark that all of our results also apply to the related problem of unit-demand min-buying pricing problem where each consumer wants to buy the cheapest products among all products that she is interested in; see \cite{GuruswamiHKKKM05} for the detailed definition of the problem.

\paragraph{Organization:} In Section~\ref{sec: bounded treewidth}, we present algorithms for graphs of bounded tree-width. We show {\sf NP}-hardness proof for {\sf GVP} on planar graphs in Section~\ref{sec: bounded genus}, and provide a polynomial time approximation scheme. We discuss the integrality gap of Sherali-Adams relaxation in Section~\ref{sec: Sherali-Adams}. We present algorithms for $k$-partite graphs and graphs of bounded degrees in Section~\ref{sec: bounded degrees}.

\section{Preliminaries}\label{sec:prelim}

\paragraph{Integral and polynomially bounded budget assumption:} We argue that we may assume w.l.o.g. that optimal prices, as well as budgets, are integral and polynomially bounded. The proof of this fact uses standard techniques and is provided for completeness. 
% and is deferred to ~\ref{sec:poly_budget_proof}.

\begin{lemma}
\label{lemma: poly budget}
Let $(G, \set{B_e}_{e \in E(G)})$ be an input instance. Then for any $\epsilon >0$, we can find, in polynomial time, another set of budgets $\set{B'_e}_{e \in E}$ such that
\begin{itemize}
  \item For all $e \in E$, $B'_e$ is integral and has value at most $P=O(m/\epsilon)$.
  \item There is a price $p'$ such that $p'(v)$ is integral for all $v$, and the revenue of $p'$ is at least $(1-\epsilon) \opt'$, where $\opt'$ is the optimal solution of the new instance.
      %\danupon{Comment: I think we have to say that revenue of $p'$ is good, not $\opt'$, because $\opt'$ might not be integral.}
  \item Any $\gamma$ approximation algorithm for an instance $(G, \set{B'_e})$ can be turned into $(1+\epsilon)\gamma$ approximation algorithm for the original instance.
\end{itemize}
\end{lemma}
\begin{proof}
Let $B_{\max}$ denote the maximum budget among all consumers, and $M= \epsilon B_{\max}/m$. Notice that $B_{\max} \leq \opt$. We create a new instance of the problem as follows: For each consumer $u v\in E$, we define new budget $B'_{e} = \floor{B_e/M}$. Let $p^*$ be the optimal price for the old instance. We define the price $p'(v) = \floor{p^*(v)/M}$.

First it is clear that the first property holds because $B_e/M \leq O(m/\epsilon)$. 
To prove the second property, let $\opt'$ denote the optimal revenue of the new instance. 
Clearly, $M \opt' \leq \opt$, since any price of new instance can be turned into the price of the old instance that collects $M$ times as much.
Let $E^* \subseteq E$ be the set of edges that contribute to the revenue w.r.t. $p^*$.  
Next, for any consumer $e=uv \in E^*$ that collects the revenue of $p^*(u)+p^*(v)$ in the old instance, we have $p'(u) + p'(v) \leq \floor{(p^*(u) + p^*(v))/M}\leq B'_e$, so this consumer does not go over budget with price $p'$ in the new instance. 
And $\sum_{e\in E^*} M (p'(u) + p'(v)) \geq \sum_{e \in E^*} \left( p^*(u) + p^*(v) - 2M \right) \geq (1-2\epsilon) \opt \geq (1-2\epsilon) M\opt'$. Therefore, the revenue collected by $p'$ is at least $(1-O(\epsilon))\opt'$.

Observe that any $\alpha$-approximation algorithm for the new instance can be turned into $(1+2\epsilon)\alpha$ approximation algorithm for the old instance as follows: Let $p'$ be the price that collects a total revenue of $\opt'/\alpha$ in the new instance. We define $p(v) = M p'(v)$ for all $v \in V$. Then the total revenue we get is $\frac{M \opt'}{\alpha} \geq \frac{(1-\epsilon)\opt}{\alpha} \geq \frac{\opt}{(1+2\epsilon)\alpha}$.
\end{proof}

\paragraph{Treewidth:} Let $G$ be a graph with treewidth at most $k$. By definition (see, e.g., \cite{RobertsonS84}), there is a {\em tree-decomposition} $(T, \cV)$ of $G$ of width $k$; that is, there is a tree $T$ and $\cV=(V_t)_{t\in T}$, a family of vertex sets $V_t\subseteq V(G)$ indexed by the vertices $t$ of $T$, with the following properties.

\begin{enumerate}
\item $V(G)=\bigcup_{t\in T} V_t$;
\item for every edge $e=uv\in G$ there exists $t\in T$ such that both $u$ and $v$ lie in $V_t$;
\item for any $v\in V(G)$, if $v\in V_{t_1}$ and $v\in V_{t_2}$ then $v\in V_{t_3}$ for any $t_3$ in the (unique) path between $t_1$ and $t_2$;
%$V_{t_1}\cap V_{t_3}\subseteq V_{t_2}$ whenever $t_1, t_2, t_3\in T$ satisfy $t_2\in t_1 T t_3$;
\item $\max_{t\in T} |V_t|=k+1$.
\end{enumerate}

For a fixed constant $k$, Bodlaender~\cite{Bodlaender96} presented a linear time algorithm that determines if the treewidth of $G$ is at most $k$ and if so constructs a corresponding tree decomposition.

\danupon{Should we move the definition of tree-width on hypergraphs here? I'm also not sure if we have to define genus and Sherali-Adams relaxation here (probably not).}

% !TEX root = main.tex

\section{{\sf FPTAS} on bounded treewidth graphs}
\label{sec: bounded treewidth}

In this section, we prove the following theorem. 

\begin{theorem} 
Let $G$ be a graph of treewidth at most $k$. 
Then there is an $(1+\epsilon)$ approximation algorithm that runs in time $\poly(|V(G)|, O(m/\epsilon)^k)$.  
\end{theorem} 

%We first show an {\sf FPTAS} on trees. Then, we show how to generalize this to {\sf FPTAS} on bounded tree-width graphs.
\iffalse
\subsection{{\sf FPTAS} on trees}\label{sec:tree-fptas}

For any node $v$ in the tree, let $C(v)$ represent the children of $v$. For any edge $uv$, let $r(uv, x)$ be the revenue at $uv$ when the sum of prices of $u$ and $v$ is $x$, i.e., $r(uv, x)=x$ if $B(uv)\geq x$ and $0$ otherwise.

The algorithm works as follows. For any node $v$ in the tree and for each integral price $P': 0 \leq P' \leq P$, we have table entry $R[v, P']$ that stores the maximum revenue we can obtain in the subtree rooted at $v$ when the price of $v$ is $P'$. $R[v, P']$ can be computed as follows. First, for any leaf node $v$, $R[v, P']=0$ for all $P' \in [P]$. For any non-leaf node $v$,
\[R[v, P']=\sum_{u\in C(v)} \max_{Q \in [P]} (R[u, Q]+f(uv, P'+Q)),
%~~~\mbox{where}~~~ f(uv, P'+Q)=
%\begin{cases}
%P'+Q & \mbox{if $P'+Q\leq B_{uv}$}\\
%0 & \mbox{otherwise}\\
%\end{cases}
\]
where $f(uv, P'+Q)$ is the revenue we received from edges between $v$ and $u$ when we price $v$ and $u$ to $P'$ and $Q$ respectively. The table entries $R[v,P']$ can be filled from the leaf nodes to root of the tree. Let $r$ be the root of the tree. We output the maximum revenue to be $\max_{P'} R[r, P']$. After we get this maximum revenue we can also find the pricing of all nodes by traversing the tree top-down.
\fi

%\subsection{{\sf FPTAS} on bounded treewidth graphs}

%We now extend the above algorithm to graphs with bounded tree-width.
We denote by $P=O(m/\epsilon)$ the maximum possible budget and prices, as obtained from Lemma~\ref{lemma: poly budget}.
Let $G$ be any input graph with edge weights satisfying Lemma~\ref{lemma: poly budget}. Now, assuming that we have a ``rooted" tree decomposition $(T, \cV)$, we solve the problem exactly in time $\poly(|T|, P^k)$ (as described below) and apply Lemma~\ref{lemma: poly budget} to get the desired {\sf FPTAS}.

First, in addition to $V_t$ for $t\in T$, we define set $E_t$ as follows. Initially, we let $E_t=\emptyset$. For each edge $e\in E(G)$, let $t(e)$ be the vertex in $T$ such that both end vertices of $e$ are in $V_t$ that is nearest to the root. Note that $t$ exists by the second property of tree decomposition (cf. Section~\ref{sec:prelim}) and is unique by the third property. We add each edge $e$ to $E_{t(e)}$. For any $t\in T$, let $T_t$ be the subtree of $T$ rooted at $t$.

For any vertex $t$ with $V_t=\{v_1, \ldots, v_{k+1}\}$ and integers $Q_1,\ldots, Q_{k+1} \in [P]$, we have table entry $R[t, Q_1, \ldots, Q_{k+1}]$ defined to be the maximum revenue we can get from edges in $\bigcup_{t'\in T_t} E_{t'}$ when the price of $v_{j}$ is $Q_j$ for all $j=1, \ldots, k+1$ (It is possible that some node $t$ may have $|V_t| <k+1$, but we assume that $|V_t| = k+1$ for simplicity of presentation. Only minor modification is required to handle the case when $|V_t| <k+1$.) Note that the size of table $R$ is $\poly(|T|, P^k)$.

%and the optimal solution is $\max_{i_1, \ldots, i_{k+1}} R[r, i_1, \ldots, i_{k+1}]$ where $r$ is the root node of $T$.

We compute $R[t, Q_1,\ldots, Q_{k+1}]$ as follows. 
When $t$ is a leaf node, we set the price of $v_{j}$ to be $Q_j$ for all $j$. Note that for all $e\in E_t$, both end vertices of $e$ are in $V_t$. Thus, we can compute the revenue we obtain from edges in $E_t$ without pricing any vertices outside $V_t$ and so we can compute $R[t, Q_1,\ldots, Q_{k+1}]$. For non-leaf node $t$, we use the following equality.
\[R[t, Q_1,\ldots, Q_{k+1}]=\sum_{t'\in C(t)}\max_{Q'_1,\ldots, Q'_{k+1}} r(t', Q'_1, \ldots, Q'_{k+1})+f(t, Q_1, \ldots, Q_{k+1})\]
where $C(t)$ is the set of children of vertex $t$ in the tree decomposition, and we define $r(t', Q'_1, \ldots, Q_{k+1}')$ and $f(t, Q_1, \ldots, Q_{k+1})$ as follows. We let $r(t', Q'_1, \ldots, Q_{k+1}')$ equals $-\infty$ if there exists ${j'}$ and $j$ such that $v_{{j'}}\in V_{t'}$ and $v_{j}\in V_t$ such that $v_{j'}=v_j$ and $i'_{j'}\neq i_j$ (in other words, the same node is set to price $i_j$ at vertex $t$ but to different price $i'_{j'}$ at vertex $t'$). Otherwise, $r(t', Q'_1, \ldots, Q_{k+1}') = R[t', Q'_1, \ldots, Q_{k+1}']$. Intuitively, we use $r(t', Q'_1, \ldots, Q_{k+1}')$ to make sure that every vertex receives only one price.
We define $f(t, Q_1, \ldots, Q_{k+1})$ to be the revenue we receive from edges in $E_t$ when we price $v_i$ to $Q_i$, for all $i$. Note that we can compute $f(t, Q_1, \ldots, Q_{k+1})$ without pricing vertices outside $V_t$ since, by definition, both end vertices of every edge in $E_t$ are in $V_t$.
%\danupon{To do: Make this clearer.} \danupon{To do: Say why we can compute $R$ this way.}

%\paragraph{Remark:} We note that this algorithm could be easily modified to give {\sf FPTAS} for bounded treewidth hypergraphs (as defined in \cite{RobertsonS84}) as well. 
%See Appendix~\ref{sec: hypergraphs} for more details.

\paragraph{Extension to bounded treewidth hypergraphs:} 
We note that the {\sf FPTAS} can be extended to solve {\sf SMP} on hypergraphs of bounded tree-width defined naturally as follows. Consider a hypergraph $H$ of width $k$ (following the definition in \cite{RobertsonS84}). This means that there is a {\em tree-decomposition} $(T, \cV)$ of $H$ of width $k$; that is, there is a tree $T$ and $\cV=(V_t)_{t\in T}$, a family of vertex sets $V_t\subseteq V(H)$ indexed by the vertices $t$ of $T$, with the following properties.

\begin{enumerate}
\item $V(H)=\bigcup_{t\in T} V_t$;
\item for every edge $e\in E(H)$ there exists $t\in T$ such that $e \subseteq V_t$;
\item for any $v\in V(H)$, if $v\in V_{t_1}$ and $v\in V_{t_2}$ then $v\in V_{t_3}$ for any $t_3$ in the (unique) path between $t_1$ and $t_2$;
%$V_{t_1}\cap V_{t_3}\subseteq V_{t_2}$ whenever $t_1, t_2, t_3\in T$ satisfy $t_2\in t_1 T t_3$;
%\item $\max_{t\in T} |V_t|=k+1$.
\end{enumerate}

It is observed that if the tree-width of $H$ is bounded by a constant, then its tree decomposition can be constructed, as follows. For any hypergraph $H$, we construct a graph $G$, called a {\em primal} graph of $H$ by letting
\[G=\left(V(H), \{(u, v) \mid u\neq v, \mbox{there exists $e\in E(H)$ such that $u, v\in e$}\}\right)\,.\]
Notice also that (see, e.g., \cite{GottlobGMSS05,AdlerGG07}) that $(T, \cV)$ is a tree decomposition of $H$ if and only if it is a tree decomposition of $G$.
Note again that we can recognize if the tree-width of $G$ is a fixed constant and, if it is, construct a tree decomposition of $G$ in linear time~\cite{Bodlaender96}. Assuming that we have the tree decomposition $(T, \cV)$, we solve the problem in time $\poly(|T|, P^k)$, where $P$ is the maximum possible budget from Lemma~\ref{lemma: poly budget}, in the same way as in Section~\ref{sec: bounded treewidth} except that we have to define $E_t$ for hyperedges instead of edges. This is done by defining $t(e)$ for each hyperedge $e$ to be the node $t$ in tree $T$ such that, among $t'\in T$ such that $V_{t'}\subseteq e$, $t$ is nearest to the root. Note that such node exists since, for any clique $C$, there is a node $t$ such that $V_t$ contains all nodes in $C$. The rest of the algorithm remains the same.

%\iffalse
%Assume without loss of generality that the smallest budget is $1$. Let $B_{nax}$ be the largest budge. Let $B(uv)$ be the budge of any edge $uv$ (assuming the graph is simple for simplicity of explanation). First, note the following lemma.
%
%\begin{lemma}\label{lem:reduce_weight}
%For any $\epsilon>0$, by losing a factor of $1/(1-O(\epsilon))$ we may assume that the price of each node is in the form $(1+\epsilon/2)^k$, for some integer $\log_{1+\epsilon/2} (\epsilon/(2n))\leq k\leq \log_{1+\epsilon/2} B_{max}$.
%\end{lemma}
%\begin{proof}(Sketch) Let $p^*$ be the optimal pricing function (i.e., $p^*(v)$ is the price of node $v$ in the optimal solution). We show that there is a pricing function $p$ that is in the form of $(1+\epsilon)^k$ such that the revenue of $p$ is at least $OPT/(1+\epsilon)$. Define $p$ as follows. First, for all $p^*(v)\leq \epsilon/(2m)$ (where $m$ is the number of edges), make $p(v)=0$. For other $v$, we set $p(v)=(1+\epsilon/2)^k$ where $k$ is the largest integer such that $(1+\epsilon/2)^k\leq p^*(v)$. Note that we lose a revenue of at most $\epsilon/2\leq \epsilon OPT/2$ (since $OPT\geq 1$) from the first step. Moreover, note that $p(v)\geq p^*(v)/(1+\epsilon/2)\geq (1-\epsilon/2)p^*(v)$ which means that $p^*(v)-p(v)\leq \epsilon p^*(v)/2$; thus, we lose at most $\epsilon OPT/2$ from the second step. Therefore, the revenue of $p$ is at least $(1-\epsilon)OPT$ as claimed.
%\end{proof}
%\fi

% !TEX root = main.tex

\section{Graphs of Bounded Genus}
\label{sec: bounded genus}

In this section we study graphs of bounded genus. First we show that the problem is strongly {\sf NP}-hard even for planar graphs ($g=0$). This implies that there is no {\sf FPTAS} for this special case. Then, we design a {\sf PTAS} for it, thus settling the complexity for the bounded genus case.

\subsection{Hardness}
\begin{theorem}\label{thm:hardness_planar}
The graph vertex pricing problem on planar graphs is strongly {\sf NP}-hard.
\end{theorem}

\begin{proof}
We show a reduction from {\sf Vertex Cover} on planar graphs, which was shown to be strongly {\sf NP}-hard in \cite{DBLP:journals/siamam/GareyJ77}. (We note that the same result can be obtained by reducing from the maximum independent set problem on planar graph, which is also strongly {\sf NP}-hard.) Assuming that we are given the instance $G=(V,E)$ of vertex cover, we construct the instance $G'$ of {\sf GVP} as follows. We add a new vertex $v'$ for each $v \in V$, and $v'$ is only connected to $v$ but not to any other vertices; it is possible to do so without violating planarity. For each edge $e \in E$, we make a copy of $e$ and add it to the instance. Call the resulting instance $G'=(V\cup V',E \cup E')$, where $\abs{V'} = \abs{V}$ and $|E'| = |V| +\abs{E}$. We will have two types of consumers: the {\em rich consumers} that correspond to (parallel) edges of $E$ have budget $|V|^2$ and $2 \abs{V}^2$ respectively, and the {\em poor consumers} corresponding to newly added edges of the form $v v'$ have budget of $1$.

The intuition for this construction is that the original edges have much more budget, so optimal solution would not try to miss any of those edges. And therefore the optimal solution would choose the vertex cover of $E$. Now we proceed to the analysis.

We let ${\sf VC}$ denote the size of minimum vertex cover of original instance $G$, and $\opt$ denote the optimal revenue of the pricing problem on the instance $G'$ constructed from $G$ using the above reduction. We claim that $\opt = 2 \abs{E} \abs{V}^2 +\abs{V} - {\sf VC}$, and once we prove this claim, we would be done. We first show that $\opt \geq 2 \abs{E} \abs{V}^2 +\abs{V} -{\sf VC}$. Let $S \subseteq V$ be the set of vertices in an optimal vertex cover of $G$. We define the following price $p$: (i) set the price $p(v) = \abs{V}^2$ and $p(v') = 0$ for each vertex $v \in S$, and (ii) $p(v) = 0$ and  $p(v') =1$ for each $v\not\in S$. Notice that we can collect $2\abs{V}^2$ from each pair of (parallel) rich consumers, resulting in a total of $2\abs{E}\abs{V}^2$, while we can collect $\abs{V} -{\sf VC}$ from the poor consumers.

For the converse, let $p$ be an optimal price for $G'$ and $C=\set{v \in V: p(v)> 1}$, i.e. $C$ is the set of vertices $v$ whose corresponding poor consumers $v v'$ do not have enough money.
\begin{claim}
Set $C$ forms a vertex cover of $G$.
\end{claim}
\begin{proof}
Suppose not. Then there is an edge $uv \in E$ not covered by any vertices in $C$. Therefore, the revenue of $uv$ is at most $p(u) + p(v) \leq 2$. The total revenue from this pricing is at most $2 (\abs{E} - 1) \abs{V}^2 + \abs{V} +4$, which is at most $2\abs{E} \abs{V}^2 - \abs{V}^2$, contradicting the fact that $p$ is an optimal price since we have already proved that $\opt \geq 2 \abs{E} \abs{V}^2$.
\end{proof}
So we know that $\abs{C} \geq {\sf VC}$, and thus the total revenue from poor consumers is at most $\abs{V}- {\sf VC}$. This implies that $\opt \leq 2 \abs{E} \abs{V}^2 + \abs{V} - {\sf VC}$, concluding Theorem~\ref{thm:hardness_planar}. 
\end{proof}

\subsection{Algorithm}
\label{sec: algorithms for bounded genus}

The main result in this section is encapsulated in the following theorem. 

\begin{theorem} 
For any fixed $H$, let $G=(V,E)$ be an $H$-minor free graph. 
% graph of genus $g$. 
Then there is a $(1+O(\epsilon))$-approximation algorithm that runs in time $|E|^{c_H}$ for some constant $c_H$ that depends on $H$. 
\end{theorem} 

It is known that any graph of genus at most $g$ is $H$-minor free for some graph $H$ such that $|V(H)|=O(g)$~\cite{Genus62}. Thus, we have the following.

%It is known that any graphs of bounded genus are minor-closed; moreover, therefore there are a constant number of forbidden minors. 

\begin{corollary} 
For any fixed $g$, let $G=(V,E)$ be graph of genus $g$. 
Then there is a $(1+O(\epsilon))$-approximation algorithm that runs in time $|E|^{O(c_g/\epsilon)}$ for some constant $c_g$ that depends on $g$. 
\end{corollary} 

The rest of this section is devoted to proving the above theorem.
The idea of the proof closely follows the standard techniques that decompose any $H$-minor free graph 
%bounded genus graphs 
into many graphs of small treewidth.  
%Let $g$ be the genus of the input graph. 
%
We will be using the following theorem, due to Demaine et.al.

\begin{theorem}(Theorem 3.1 in~\cite{DemaineHK05})
For a fixed $H$, there is a constant $c_H$ such that, for any integer $k \geq 1$, and for every $H$-minor free graph $G$, the vertices of $G$ can be partitioned into $k+1$ sets such that any $k$ of the sets induce a graph of treewidth at most $c_H k$. Furthermore, such a partition can be found in polynomial time.
\end{theorem}

We choose the parameter $k= \ceil{1/\epsilon}$ and invoke the theorem to partition the vertex set $V(G)$ into $V(G) = \bigcup_{i=1}^k V_i$. We then create $k$ instances $H_1,\ldots, H_k$ where instance $H_j$ is obtained by removing vertices in $V_j$ and their adjacent edges from $G$. From the theorem, each graph $H_j$ has treewidth at most $O(1/\epsilon)$. 
The following claim asserts that one of these subgraphs $H_j$ admits a near-optimal pricing solution. 

\begin{claim}
$\sum_{j=1}^k \opt(H_j) \geq (k-2) \opt$ where $\opt(H_j)$ denotes the optimal value for the instance $H_j$. 
\end{claim}

Before proving the claim, we argue that it does imply the PTAS. 
Let $j^*$ be the index such that $\opt(H_{j^*})$ is maximized, so we have $\opt(H_{j^*}) \geq (1-O(\epsilon)) \opt$. Since the treewidth of $H_{j^*}$ is at most $O(1/\epsilon)$, we can compute the $(1+\epsilon)$-approximate pricing in $H_{j^*}$ in time $m^{O(1/\epsilon)}$ by applying Lemma~\ref{lemma: poly budget} and the algorithm mentioned in Section~\ref{sec: bounded treewidth}.

\begin{proof}
Let $p^*$ be an optimal price for graph $G$ that collects the revenue of $\opt$. For any edge $e=uv$, let $p^*(e)=p^*(u)+p^*(v)$. Let $E^* \subseteq E(G)$ be the subset of edges $e$ such that $p^*(e) \leq B_e$. 
These are the edges that have positive contribution to the revenue $\opt$. 
For each subset of vertices $S \subseteq V$, denote by $\delta_{E^*}(S)$ the set of edges in $E^*$ with exactly one endpoint in $S$, and $E^*(S)$ the set of edges in $E^*$ with both endpoints in $S$. 

For each graph $H_j$, if we set the price $p^*$ (i.e. the optimal price induced on $H_j$), we would be able to collect the revenue of at least $r_j = \sum_{i \neq j} \sum_{e \in E^*(V_i) } p^*(e)  + \sum_{ e \in \tilde E \setminus \delta_{E^*}(V_j) } p^*(e) $, where $\tilde E$ denotes the set of edges connecting two vertices in different sets $V_i$. 
Summing over all $j$, we argue that $\sum_{j=1}^k r_j \geq (k-2) \opt$: Notice that each edge in $E^*(V_j)$ contributes exactly $k-1$ times in the sum for each $j$, while each edge in $\delta_{E^*}(V_j)$ contributes exactly $k-2$ times (edge $uv$ connecting $V_i$ to $V_{i'}$ contributes to all terms except for $r_i$ and $r_{i'}$). 
\end{proof}

\section{Integrality gap of the Sherali-Adams relaxation on bounded treewidth and bounded genus graphs}
\label{sec: Sherali-Adams}
We describe a family of LP relaxation, denoted by (LP-$r$), and a rounding algorithm that computes an optimal solution given an optimal fractional solution for (LP-$r$), provided that the treewidth of the input graph is at most $r/2$. One can use a standard procedure to check that the constraints of (LP-$r$) can be generated by applying $O(r)$ rounds of Sherali-Adams hierarchy on a natural LP relaxation. 
For completeness, we provide the LP description and the proof that Sherali-Adams can generate (LP-$r$) in Appendix~\ref{appendix: sherali-adams}.

\paragraph{Terms and notation:} Let $P$ be the maximum possible price obtained from Lemma~\ref{lemma: poly budget}. Given an input graph $G=(V,E)$, let $p,p': V \rightarrow [P]$ (or equivalently $p,p' \in [P]^V$) be price functions that assign prices to vertices in graph $G$. We say that two functions $p, p'$ {\em agree} on $S$ if and only if $p(v) = p'(v)$ for all $v \in S$. For any subset of vertices $S \subseteq V$, a {\em restriction} of $p$ on set $S$, denoted by $p|_S$, is the (unique) function $\tilde{p} \in [P]^S$ that agrees with $p$ on $S$.

\paragraph{LP relaxation:} For each set $S \subseteq V$, and assignment $\alpha \in [P]^S$, we introduce an LP variable $y(S,\alpha)$ which is supposed to be an indicator that $p$ agrees with $\alpha$ on $S$, i.e. $y(S,\alpha) =1$ if the solution function $p$ assigns the value $p(v) = \alpha(v)$ for all $v \in S$, and $y(S, \alpha) = 0$ otherwise. For any two assignments $\alpha \in [P]^X$ and $\beta \in [P]^Y$ such that $X$ and $Y$ are disjoint, we write $\alpha \cup \beta$ to represent the assignment $\gamma \in [P]^{X \cup Y}$ that agrees with $\alpha$ on $X$ and with $\beta$ on $Y$. We use the following LP relaxation with $r \geq 2$.
\noindent\begin{align}
\label{eq:LP}\tag{LP-$r$}
\footnotesize
\hspace{-5pt}
\noindent\begin{aligned}
%\mbox{(LP-$r$)}\\
\lefteqn{\max \quad \sum_{e=(u,v) \in E} \sum_{\alpha \in [P]^{\set{u,v}}: \atop \alpha(u) + \alpha(v) \leq B_e} (\alpha(u) + \alpha(v)) y(\set{u,v}, \alpha)}\hphantom{\max}  \\
{\rm s.t.} && \sum_{\beta \in [P]^{T}}  y(S \cup T, \alpha \cup \beta)  &=  y(S, \alpha)  &&   \forall S,T \subseteq V: |S \cup T| \leq r, S \cap T = \emptyset, \alpha \in [P]^S\\
&& 0 \leq y(S, \alpha)  &\leq 1 &&\forall S: \abs{S} \leq r, \forall \alpha \in [P]^S\\
&& y(\emptyset, \emptyset) &= 1
\end{aligned}
\end{align}
%
%\begin{eqnarray*}
%\mbox{(LP-$r$)} &  & \\
% & \max & \sum_{e=(u,v) \in E} \sum_{\alpha \in [P]^{\set{u,v}}: \atop \alpha(u) + \alpha(v) \leq B_e} (\alpha(u) + \alpha(v)) y(\set{u,v}, \alpha) \\
% & \mbox{s.t.} & \sum_{\beta \in [P]^{T}}  y(S \cup T, \alpha \cup \beta)  =  y(S, \alpha) ~~~ \forall S,T \subseteq V: |S \cup T| \leq r, S \cap T = \emptyset \\
% & & y(\emptyset, \emptyset) = 1\\
% & & 0 \leq y(S, \alpha)  \leq 1 ~~~ \forall S: \abs{S} \leq r, \forall \alpha \in [P]^S
%\end{eqnarray*}
%
The size of the LP is $P^{O(r)}$. A natural way to view this LP is to treat the variables $\set{y(S, \alpha)}_{\alpha}$ for a fixed set $S \subseteq V$ as a probability distribution where the value of $y(S, \alpha)$ represents the probability that the vertices in set $S$ are assigned price $\alpha$. Notice that we have the constraint $1=y(\emptyset, \emptyset) =  \sum_{\alpha \in [P]^S} y(S, \alpha)$.

We first argue that this LP is indeed a relaxation for {\sf GVP}. Let $p^*$ be an optimal (integral) price function. For any set $S \subseteq V$ and any assignment $\alpha$ for $S$, we assign $y(S,\alpha) =1$ if and only if $p^*(v) = \alpha(v)$ for all $v \in S$. Consider any two subsets $S, T: S \cap T = \emptyset$ and any assignment $\alpha \in [P]^S$, and notice that $y(S, \alpha) = 1$ if and only if there exist $\beta \in [P]^T$ such that $y(S \cup T, \alpha \cup \beta) = 1$. Therefore, the above solution satisfies all constraints, and the objective value equals to the total revenue collected by $p^*$.

Now let $G$ be an input graph with treewidth $k$. First, we solve (LP-$(k+1)$) and denote the optimal objective value by $\opt$. We describe below a randomized algorithm that returns price function $p$ with expected total profit of $\opt$.

\paragraph{Algorithm description:} Denote by $(T, \cV)$ the tree decomposition of graph $G$. Initially we have price function $p$ where $p(v)$ is undefined for all $v \in V$. We process each element of the tree $t \in T$ in order defined by the tree $T$ from root to leaves, ensuring that whenever any node $t \in T$ is processed, all ancestors of $t$ have already been processed. When the algorithm processes the node $t \in T$, it assigns the prices to vertices in $V_t$, defining the values $p(v)$ for all $v \in V_t$. It is clear that function $p$ will be defined for all $v \in V$ in the end.

%Once the algorithm finishes processing any set $V_{t}$, it assigns the prices to all vertices in $V_t$ in such a way that the price assignments of $t$ and its parent $t'$ agree on $V_t\cap V_{t'}$.

Let $t \in T$ be the current tree node that is being processed. If $t$ is the root of the tree, we assign the prices to vertices in $V_t$ according to the probability distribution $\set{y(V_t, \alpha)}_{\alpha \in [P]^{V_t}}$. Otherwise, let $t'$ be the parent of $t$. Suppose $\beta \in [P]^{V_{t'}}$ is a price assignment of vertices $V_{t'}$. Define $X'= V_t \setminus V_{t'}$ to be the set of vertices in $V_t$ whose prices have not been assigned. We pick a random assignment $\alpha \in [P]^{X'}$ using the distribution $y(V_{t'} \cup X', \beta \cup \alpha)/y(V_{t'}, \beta)$, i.e. for each assignment $\alpha' \in [P]^{X'}$, $\pr{}{\alpha = \alpha'} = y(V_{t'} \cup X', \beta \cup \alpha')/y(V_{t'}, \beta)$. Notice that this is a valid probability distribution because of the constraint
\[\sum_{\alpha' \in [P]^{X'}} y(V_{t'} \cup X', \beta \cup \alpha') = y(V_{t'}, \beta)\]

In the end of the algorithm, the following property holds.

\begin{lemma}
\label{lemma: sherali-adams rounding}
For any $t \in T$, the price $p \mid_{V_t}$ has the same distribution as $\set{y(V_t, \alpha)}_{\alpha \in [P]^{V_t}}$. In other words, for each tree node $t \in T$, we have
\[(\forall \alpha \in [P]^{V_t}) \mbox{  } \pr{}{p|_{V_t} \mbox{ agrees with } \alpha } = y(V_t, \alpha) \]
\end{lemma}

\begin{proof}
We prove by induction on the ordering of tree nodes by tree $T$. For the root of the tree, the probability that $p|_{V_t}$ equals $\alpha$ is exactly $y(V_t,\alpha)$. Now we consider any non-root tree node $t$, and assume that the lemma holds for the parent tree node $t' \in T$. For any assignment $\alpha \in [P]^{V_t}$, we have
\begin{eqnarray*}
\pr{}{p\mid_{V_t} \mbox{ agrees with } \alpha } &=& \sum_{\beta \in [P]^{V_{t'} \cap V_{t}}}  \pr{}{p \mid_{V_t \cap V_{t'}} = \beta } \pr{}{p \mbox{ agrees with } \alpha  \mid p|_{V_{t'}\cap V_t} = \beta }   \\
&=& \sum_{\beta}  y(V_t \cap V_{t'},\beta) \frac{y(V_t , \alpha )}{y(V_{t'}\cap V_t,\beta)} \\
&=& \sum_{\beta} y(V_t, \alpha ) \\
&=& y(V_t, \alpha)
\end{eqnarray*}

Note that the second line follows from the first line by the induction hypothesis. The rest is simply a calculation.
\end{proof}

Now assuming the lemma, we have the following corollary. 
\begin{corollary} 
The expected profit collected by the algorithm is at least $\opt$.  
\end{corollary} 
\begin{proof} 
For each customer edge $uv \in E$, the expected profit made by $uv$ equals to
\[\sum_{\alpha: \alpha(u)+\alpha(v) \leq B_e} (\alpha(u) + \alpha(v)) \pr{}{p \mbox{ agrees with } \alpha} \]
Since both end vertices of each edge $e=(u, v)$ belong to some set $V_t$, from the Lemma~\ref{lemma: sherali-adams rounding}, the distribution of $p|_{\set{u,v}}$ is the same as that of $\set{y(\set{u,v}, \alpha)}_{\alpha \in [P]^{\set{u,v}}}$. So the probability term can be replaced by $y(\set{u,v}, \alpha)$.
\end{proof}

\paragraph{Derandomization:} We note that the above algorithm can be derandomized by the method of conditional expectation. Alternatively, it can also be easily derandomized using the fact that {\em any} solution can be obtained from the above randomized algorithm with positive probability will give a profit $\opt$. This is due to the above corollary which says that the random assignment chosen by the algorithm is optimal in expectation, but no assignment is better than optimal; therefore, any assignment chosen with positive probability will be optimal. Thus, we can just propagate the assignments through the tree decomposition using any assignment in the support of the LP (or equivalently, in the support of the distribution suggested by the rounding).

\paragraph{Bounded genus graphs:} One can show that the integrality gap of (LP-$r$) is at most $1+O(\epsilon)$ for $r= O(g/\epsilon)$. The proof follows along the same line as in Section~\ref{sec: algorithms for bounded genus}, but we work with LP solution instead. We sketch it here for completeness.
First, we need the following theorem, which is a generalization of the Sherali-Adams rounding algorithm presented earlier, whose proof can be obtained by a trivial modification.   

\begin{theorem}
\label{thm:SA-rounding}
Let $G$ be any graph and $\set{y(S,\alpha)}$ be a feasible solution for (LP-$r$) on $G$. 
Also, let $G'$ be a subgraph of $G$ such that $G'$ has treewidth at most $r-1$.
Then, there is a polynomial time algorithm that collects the revenue of at least $\opt_{LP-r}(G')$.  
\end{theorem} 

Now we present the integrality gap upper bound. 
Suppose we have an optimal LP solution $\set{y(S,\alpha)}_{S, \alpha}$ of (LP-$r$), and let $\opt_{LP-r}$ denote the LP-cost of this solution. 
For any set of edges $E'$, let $\opt_{LP-r}(E')$ denote the LP-cost of this solution we get from edges in $E'$; i.e., 
\[
\opt_{LP-r}(E') = \sum_{e=(u,v) \in E'} \sum_{\alpha \in [P]^{\set{u,v}}: \atop \alpha(u) + \alpha(v) \leq B_e} (\alpha(u) + \alpha(v)) y(\set{u,v}, \alpha)
\]
For any subgraph $G'$ of $G$, we also let $\opt_{LP-r}(G') = \opt_{LP-r}(E(G'))$. Using a technique similar to Section~\ref{sec: algorithms for bounded genus}, we will show that there is a subgraph $G'$ of $G$ such that 
\[
\opt_{LP-r}(G') \geq (1-O(\epsilon)) \opt_{LP-r} \geq \opt_{LP-r}/(1+O(\epsilon))\,.
\]
Let $K= 1/\epsilon$. 
By the same method as in Section~\ref{sec: bounded genus}, we partition edges of $G$ into $E_1, \ldots, E_K$ and, for any $i$, let $G_i$ be the subgraph of $G$ such that $E(G_i)=E(G)\setminus E_i$. We guarantee that $G_i$ has small treewidth. 
%
%We create $K$ sub-instances $G_1,\ldots, G_K$ by the same method as in Section~\ref{sec: bounded genus}. 
%
Notice that instance $G_i$ has 
\[\opt_{LP-r}(G_i) = \opt_{LP-r} - \opt_{LP-r}(E_i),\] 
so there must be instance $G_{i^*}$ with total LP-cost of 
\[
\opt_{LP-r}(G_{i^*})\geq (1-1/K)\opt_{LP-r}=(1-O(\epsilon))\opt_{LP-r}.
\]
Since $G_{i^*}$ has small treewidth, i.e. at most $k=O(g/ \epsilon)$, if we ensure that $r$ is at least $k+1$, 
we can apply Theorem~\ref{thm:SA-rounding} to collect the revenue of $\opt_{LP-r}(G_{i^*}) \geq (1-O(\epsilon)) \opt_{LP-r}$, thus bounding the integrality gap of (LP-$r$) by a factor of $(1+O(\epsilon))$ as claimed.  

%By solving sub-instance $G_{i^*}$ using the LP rounding algorithm for bounded treewidth instances, we will get a pricing that collects revenue of $(1-\epsilon)\opt_{LP-r}$ (in Section~\ref{sec: bounded genus}, we need to solve all sub-instances and return the one with maximum revenue). So the integrality gap is at most $1/(1-\epsilon) \leq (1+2 \epsilon)$.

%Let $rev(E_i)$ be the total LP-cost collected from edges in $E_i$. Notice that instance $G_i$ has total LP-cost of $\opt_{LP-r} - rev(E_i)$, so there must be instance $G_{i^*}$ with total LP-cost of $(1-\epsilon)\opt_{LP-r}$. By solving sub-instance $G_{i^*}$ using the LP rounding algorithm for bounded treewidth instances, we will get a pricing that collects revenue of $(1-\epsilon)\opt_{LP-r}$ (in Section~\ref{sec: bounded genus}, we need to solve all sub-instances and return the one with maximum revenue). So the integrality gap is at most $1/(1-\epsilon) \leq (1+2 \epsilon)$.

%\paragraph{Sherali-Adams relaxation:} It is a standard fact that variables of the form $y(S, \alpha)$ in (LP-$r$) can be generated by $r$ rounds of Sherali-Adams lift-and-project operations. See Appendix~\ref{appendix: sherali-adams} for more detail.
\danupon{I removed this paragraph since we already talked about it earlier.}

% !TEX root = main.tex

\section{$k$-partite graphs and bounded-degree graphs}
\label{sec: bounded degrees}

We note that, by using the same arguments as in Theorem~\ref{thm:hardness_planar}, one can show that {\sf GVP} is APX-hard even on graphs of constant degrees and $k$-partite graphs when $k$ is constant: By using the same reduction except that now we start from {\sf Vertex Cover} on cubic graphs (which is APX-hard), we can show that {\sf GVP} is APX-hard as well, and the degree of each vertex in the resulting instance $G'$ is at most a constant. 
So there is no PTAS, unless P=NP, even in bounded-degree graphs, and in this section we give new algorithmic results for many special cases. 
Our results are summarized in the following theorem. 

\begin{theorem} 
{\sf GVP} is polynomial time solvable on graphs of degree at most two and admits a $2$-approximation algorithm on graphs of degrees at most four. 
Moreover, in $k$-partite graphs, there is a $4(1- 1/k)$ approximation algorithm. 
\end{theorem} 

\subsection{Improved approximation algorithms on $k$-partite graphs}\label{sec:kcolor}

Let $G$ be a $k$-partite graph. We assume that we know the partition $G$ into $k$ partitions. To avoid confusion, we sometimes called each partition a {\em color class}. First we randomly partition the color classes into two roughly equal sides called $L$ and $R$, i.e. we ensure that the number of color classes in $L$ differ by those in $R$ by at most one. We say that an edge $e$ is cut if exactly one endpoint of $e$ belongs to $L$. We analyze this algorithm in two cases using counting arguments. For the sake of analysis, we think of each node as having a unique integer ID. We need the following lemma from \cite{BalcanB07}.

\begin{lemma}\label{lem:Balcan-Blum}
(Balcan and Blum \cite{BalcanB07}) There is a 2-approximation for the Graph Vertex Pricing problem on bipartite graphs.
\end{lemma}

%Let the vertices of $G$ be numbered from $[n]$.

%We will refer to these numbers as node IDs.
% Question: How to do this? Can we derandomize?

\paragraph{$k$ is even:} Notice that the number of possible cuts is ${k-1 \choose k/2 -1}$. Each edge $e$ belongs to exactly $\binom{k-2}{k/2-1}$ cuts (after fixing two end vertices of the edge, we pick $k/2-1$ more color classes to join the side of the node with smaller ID). Therefore, the probability that each edge is cut is
\[\frac{{k-2 \choose k/2- 1}}{{k-1 \choose k/2-1}} = k/2(k-1)\,.\]
Observe that this term is slightly more than $1/2$. Using Lemma~\ref{lem:Balcan-Blum}, we get an approximation ratio of $4(k-1)/k$.

%$$(x+y)/x= 1+y/x = 1+\frac{(k/2-1)!(k/2-1)!}{(k/2-2)!(k/2)!}=1+\frac{k/2-1}{k/2}=2\frac{k-1}{k}$$

%Edge will not be in the cut $y=\binom{k-2}{k/2-2}$ (after fixing two end vertices of the edge on one side, we pick $k/2-2$ more vertices to join that side). Note that,

%Therefore, the probability that an edge is in the cut is $\frac{k}{2(k-1)}$. Using the algorithm of Balcan and Blum \cite{BalcanB07} on bipartite graphs, we get an approximation ratio of $4(k-1)/k$.

\paragraph{$k$ is odd:} In this case, the number of possible cuts is $\binom{k}{\lfloor k/2\rfloor}$. Each edge will be in  $2\binom{k-2}{\lfloor k/2\rfloor-1}$ cuts (after fixing two end vertices of the edge, we pick $\lfloor k/2\rfloor-1$ more color classes to join one side out of two sides). It is easy to see that this term is equal to $\binom{k-2}{\lfloor k/2\rfloor-1} + \binom{k-2}{\lfloor k/2\rfloor} =\binom{k-1}{\lfloor k/2\rfloor}$. Hence the probability that each edge is cut is
$\frac{\binom{k-1}{\lfloor k/2\rfloor}}{\binom{k}{\lfloor k/2\rfloor}} = \frac{k+1}{2k}\,.$
Again, using Lemma~\ref{lem:Balcan-Blum}, we get an approximation ratio of $\frac{4k}{k+1}$.
%times (after fixing two end vertices of the edge, we pick $\lfloor k/2\rfloor-1$ more vertices to join one side out of two sides). Edge will not be in the cut $y=\binom{k-2}{\lfloor k/2\rfloor-2}+\binom{k-2}{\lceil k/2\rceil-2}$ times (after fixing two end vertices of the edge on one side, we pick $\lfloor k/2\rfloor-2$ {\em or} $\lceil k/2\rceil-2$ more vertices to join that side). Note that,

%$$(x+y)/x= 1+y/x = 1+\frac{(\lfloor k/2\rfloor-1)!(\lceil k/2\rceil -1)!}{(k/2-2)!(k/2)!}=???$$
%$$(x+y)/x= 1+y/x = 1+\frac{\lfloor k/2\rfloor}{\lceil k/2 \rceil}=\frac{2\lceil k/2\rceil-1}{\lceil k/2 \rceil} =2\frac{k}{k+1}$$
%
%Therefore, the probability that an edge is in the cut is $\frac{k+1}{2k}$.

In general graphs, we may think that $k=n$, so this will give a slightly improved approximation factor of $4(1-1/n)$.
%
%
%Therefore, the probability that an edge is in the cut is $\frac{k}{2(k-1)}$. Using the standard algorithm on bipartite graphs, we get an approximation ratio of $???$.
%
\paragraph{Deterministic algorithm:} We use the method of conditional expectation. Put an arbitrary color class on the left side $L$. For each color class, determine whether we should put it in $L$ or $R$ as follows. Put the color class on the left side and calculate the conditional expectation (which can be computed since we can compute a probability that an edge is in the cut which depends on the number of vertices already present in $L$ and $R$). Try putting the color class in $R$ and do the same. Pick the choice that gives better expected revenue.

\paragraph{Tightness of the algorithm:} Our approximation factor here is, in some sense, tight: We argue that improving this bound further for any even number $k$ would immediately imply an improved approximation ratio for general graphs. Suppose we have an $\alpha$ approximation algorithm for $k$-partite graphs where $\alpha < 4(k-1)/k$. We randomly partition nodes into $k$ parts and delete edges that connect two vertices in the same set of the partition. Each edge is deleted with probability $1/k$. Then we run the $\alpha$-approximation algorithm for the resulting $k$-partite graph, and this would give us an approximation ratio of $\alpha k/(k-1) < 4$ for general graphs.

\subsection{Polynomial-time algorithms for Graphs of degree at most two}\label{sec:degree2}

We observe that when the input graph has degree at most two, i.e. when in consists of paths and cycle, the problem can be solved in polynomial time. This can be done by considering two cases: when every edge contributes a non-zero revenue to the optimal solution and otherwise. In the former case, we can solve the problem by writing a linear program. In the latter case, the problem is solved by a dynamic programming technique.

%Due to the lack of space, details can be found in Appendix~\ref{sec:path_cycle_detail}.

We observe that for any graph $G(V,E)$ if every edge $e \in E$ contributes a positive amount to the total revenue of the optimal solution then the optimal vertex prices can be calculated using the following linear program :

\begin{eqnarray*}
\mbox{($LP_{opt}$)} & & \\
 & \max & \sum_{e=uv \in E}{(p(u)+p(v))} \\
 & \mbox{s.t.} & p(u) + p(v) \leq B_e \\
 & & p(u) \geq 0\ \mbox{for all}\ u \in V \\
\end{eqnarray*}

Let $LP_{opt}{(G)}$ represent the optimal value of $LP_{opt}$ on a graph $G$. Let $P =\{v_1,v_2,\dots,v_n\}$ be a path with edges $(v_i,v_{i+1})$ for $1 \leq i \leq n-1$. Let $OPT(G)$ represent the optimal revenue that can be obtained from the graph $G$. For $i < j$, let $R[v_i,v_j]$ be the maximum revenue that can be obtained from the subpath (say $P_{ij}$) induced by the vertices $v_i, v_{i+1},\dots,v_j$. Note that $R[v_i,v_i]=0$ for all $1 \leq i \leq n$. For $i < j$, $R[i,j]$ can be computed using the following recursion.

\[
  R[i,j] = \max \left\{
  \begin{array}{l l}
    LP_{opt}(P_{ij}) & \quad \text{if all the edges of $P_{ij}$}\\
     & \quad \text{contribute positive revenue.}\\
    \max_{k=i,\dots,j-1} R[v_i,v_k] + R[v_{k+1},v_j] & \quad \text{otherwise} \\
  \end{array} \right.
\]

$OPT(P) = R[1,n]$ is the optimal revenue obtained from the path $P$. By solving the above recursion we obtain a polynomial time algorithm for paths.

Now we design a polynomial time algorithm for cycles. Let $C={v_1,v_2,\dots,v_n}$ be a cycle with edges $(v_1,v_2),(v_2,v_3),\dots$,$(v_{n-1},v_n),(v_n,v_1)$. If all the edges of $C$ contribute positive revenue then the maximum revenue is given by $LP_{opt}(C)$. If one of the edges $uv \in C$ contributes zero to the total revenue then the maximum revenue is obtained by using the above mentioned algorithm on the path obtained by deleting $uv$ from $C$. Hence we get the following recursion.

\[
  OPT(C) = \max \left\{
  \begin{array}{l l}
    LP_{opt}(C) & \quad \text{if all the edges of $C$}\\
     & \quad \text{contribute positive revenue}\\
    \max_{e \in V(C)} OPT(C \setminus e) & \quad \text{otherwise} \\
  \end{array} \right.
\]

Since $C \setminus e$ is a path, we use the previous algorithm to compute $OPT(C \setminus e)$ in polynomial time. Solving this recursion we get a polynomial time algorithm for cycles.

Any graph of degree at most two consists of paths and cycles. By combining the above mentioned algorithms for paths and cycles we get a polynomial time algorithm for graphs of degree at most two.

%\iffalse
%\subsection{Graphs of degree at most three}
%
%By using the algorithm mentioned in Section \ref{sec:kcolor}, we get a factor $3$ approximation on graphs with degree at most three. Now we present a better algorithm achieving an approximation factor of $2+\epsilon$.
%
%Let $G$ be a graph with degree at most three. Let $C$ be any maximal set of edge-disjoint cycles in $G$. Let $F$ be the graph obtained by deleting the {\em{edges}} of $C$ from $G$. Since the degree of every vertex in $G$ is at most three, $F$ is a forest. Using the algorithm from Section \ref{sec:degree2} we can compute $OPT(C)$ in polynomial time. Using the algorithm in Section \ref{sec:tree-fptas} we can approximate $OPT(F)$ within a factor of $(1+\epsilon)$. Note that $OPT(G) \leq OPT(C) + OPT(F)$. Hence $\max \{ OPT(C), OPT(F) \}$ approximates $OPT(G)$ within a factor of $2+\epsilon$.
%%By using the algorithm of [ABC], $C$ can be obtained from $G$ in polynomial time.
%\fi

\subsection{Graphs of degree at most four}

By using the algorithm mentioned in Section \ref{sec:kcolor}, we get factor $3$ approximation algorithms on graphs with degree at most three and four. Now we present a better algorithm achieving an approximation factor of $2$ for graphs of degree at most four.

Let $G=(V,E)$ be a graph with degree at most four. We decompose edges of $G$ into $E= E_1 \cup E_2$ such that the corresponding subgraphs $G_1 = (V,E_1)$ and $G_2 = (V,E_2)$ have degree at most two: This is a standard trick. First add an arbitrary matching $M$ between odd-degree vertices, so the resulting graph is Eulerian. Let $C$ be its Eulerian tour. Then we let $E_1= C \setminus M$ and $E_2 = E \setminus E_1$.

Using the algorithm from Section \ref{sec:degree2} we can compute optimal solutions to $G_1$ and $G_2$ in polynomial time. Note that $\opt(G_1) + \opt(G_2) \geq \opt(G)$, so we get $2$-approximation.

%\section{{\sf NP}-Hardness}

\section{Open Problems}

This paper settled the complexity of {\sf GVP} in graphs of bounded genus. Since the integrality gap of Sherali-Adams hierarchy applied to a natural LP for these cases is $(1+\epsilon)$, it is interesting to see how this integrality gap behaves in general graphs. Using lift-and-project LP might be a possible way to get improved approximation algorithms or better hardness results. Note that, for two rounds of Sherali-Adams hierarchy, the integrality gap is at least $4-\epsilon$ by Khandekar et al. We also considered the bounded degree cases and show a $2$-approximation algorithm for cubic graphs. It is also interesting to see whether this ratio is tight. We believe that this problem remains {\sf NP}-hard even on trees.

%Is the approximation ratio of two on cubic graphs tight? Sherali-Adams relaxation of the general case? Conjecture: NP-hard on bounded tree-width graphs. Stronger Conjecture: NP-hard on trees.

\paragraph{Acknowledgement:} We would like to thank the reviewers for several thoughtful comments.

%\vspace{0.15in}
\bibliographystyle{alpha}
\bibliography{bib-main}

\newcommand{\etalchar}[1]{$^{#1}$}
\begin{thebibliography}{GGM{\etalchar{+}}05}

\bibitem[AGG07]{AdlerGG07}
Isolde Adler, Georg Gottlob, and Martin Grohe.
\newblock {Hypertree Width and Related Hypergraph Invariants}.
\newblock {\em Eur. J. Comb.}, 28(8):2167--2181, 2007.

\bibitem[Bak94]{Baker94}
Brenda~S. Baker.
\newblock {Approximation Algorithms for NP-Complete Problems on Planar Graphs}.
\newblock {\em J. ACM}, 41(1):153--180, 1994.

\bibitem[BB07]{BalcanB07}
Maria-Florina Balcan and Avrim Blum.
\newblock {Approximation Algorithms and Online Mechanisms for Item Pricing}.
\newblock {\em Theory of Computing}, 3(1):179--195, 2007.
\newblock Also in EC'06.

\bibitem[BHKY62]{Genus62}
Joseph Battle, Frank Harary, Yukihiro Kodama, and J.~W.~T. Youngs.
\newblock Additivity of the genus of a graph.
\newblock {\em Bull. Amer. Math. Soc.}, 68:565--568, 1962.

\bibitem[BK11]{BriestK11}
Patrick Briest and Piotr Krysta.
\newblock Buying cheap is expensive: Approximability of combinatorial pricing
  problems.
\newblock {\em SIAM J. Comput.}, 40(6):1554--1586, 2011.

\bibitem[B{\"O}04]{Treewidth-SheraliAdams}
Daniel Bienstock and Nuri {\"O}zbay.
\newblock {Tree-width and the Sherali-Adams Operator}.
\newblock {\em Discrete Optimization}, 1(1):13--21, 2004.

\bibitem[Bod96]{Bodlaender96}
Hans~L. Bodlaender.
\newblock {A Linear-Time Algorithm for Finding Tree-Decompositions of Small
  Treewidth}.
\newblock {\em SIAM J. Comput.}, 25(6):1305--1317, 1996.
\newblock Also in STOC'93.

\bibitem[Bro41]{Brook41}
R.~L. Brooks.
\newblock {On Colouring the Nodes of a Network}.
\newblock {\em Proc. Cambridge Philosophical Society, Math. Phys. Sci.},
  37:194--197, 1941.

\bibitem[CCKK12]{ChalermsookCKK12}
Parinya Chalermsook, Julia Chuzhoy, Sampath Kannan, and Sanjeev Khanna.
\newblock Improved hardness results for profit maximization pricing problems
  with unlimited supply.
\newblock In Anupam Gupta, Klaus Jansen, Jos{\'e} D.~P. Rolim, and Rocco~A.
  Servedio, editors, {\em APPROX-RANDOM}, volume 7408 of {\em Lecture Notes in
  Computer Science}, pages 73--84. Springer, 2012.

\bibitem[CDF{\etalchar{+}}09]{CardinalDFJNW09}
Jean Cardinal, Erik~D. Demaine, Samuel Fiorini, Gwena{\"e}l Joret, Ilan Newman,
  and Oren Weimann.
\newblock {The Stackelberg Minimum Spanning Tree Game on Planar and
  Bounded-Treewidth Graphs}.
\newblock In {\em WINE}, pages 125--136, 2009.

\bibitem[CKR10]{SparsestCut-Treewidth}
Eden Chlamtac, Robert Krauthgamer, and Prasad Raghavendra.
\newblock {Approximating Sparsest Cut in Graphs of Bounded Treewidth}.
\newblock In {\em APPROX-RANDOM}, pages 124--137, 2010.

\bibitem[CLN13a]{ChalermsookLN13soda}
Parinya Chalermsook, Bundit Laekhanukit, and Danupon Nanongkai.
\newblock Graph products revisited: Tight approximation hardness of induced
  matching, poset dimension and more.
\newblock In Sanjeev Khanna, editor, {\em SODA}, pages 1557--1576. SIAM, 2013.

\bibitem[CLN13b]{ChalermsookLN13FOCS}
Parinya Chalermsook, Bundit Laekhanukit, and Danupon Nanongkai.
\newblock Independent set, induced matching, and pricing: Connections and tight
  (subexponential time) approximation hardnesses.
\newblock In {\em FOCS}, 2013.

\bibitem[CT12]{hierarchies-survey}
Eden Chlamtac and Madhur Tulsiani.
\newblock {Convex Relaxations and Integrality Gaps}.
\newblock {\em Handbook on Semidefinite, Conic and Polynomial Optimization,
  International Series in Operations Research and Management Science},
  166:139--169, 2012.

\bibitem[DHiK05]{DemaineHK05}
Erik~D. Demaine, Mohammad~Taghi Hajiaghayi, and Ken ichi Kawarabayashi.
\newblock {Algorithmic Graph Minor Theory: Decomposition, Approximation, and
  Coloring}.
\newblock In {\em FOCS}, pages 637--646, 2005.

\bibitem[Epp00]{Eppstein00}
David Eppstein.
\newblock {Diameter and Treewidth in Minor-Closed Graph Families}.
\newblock {\em Algorithmica}, 27(3):275--291, 2000.
\newblock Also in SODA'95.

\bibitem[GGM{\etalchar{+}}05]{GottlobGMSS05}
Georg Gottlob, Martin Grohe, Nysret Musliu, Marko Samer, and Francesco
  Scarcello.
\newblock {Hypertree Decompositions: Structure, Algorithms, and Applications}.
\newblock In {\em WG}, pages 1--15, 2005.

\bibitem[GHK{\etalchar{+}}05]{GuruswamiHKKKM05}
Venkatesan Guruswami, Jason~D. Hartline, Anna~R. Karlin, David Kempe, Claire
  Kenyon, and Frank McSherry.
\newblock {On Profit-Maximizing Envy-free Pricing}.
\newblock In {\em SODA}, pages 1164--1173, 2005.

\bibitem[GJ77]{DBLP:journals/siamam/GareyJ77}
M.~R. Garey and David~S. Johnson.
\newblock {The Rectilinear Steiner Tree Problem in NP Complete}.
\newblock {\em SIAM Journal of Applied Mathematics}, 32:826--834, 1977.

\bibitem[Hea90]{heawood90}
P.~J. Heawood.
\newblock {Map Colour Theorem}.
\newblock {\em Quarterly Journal of Pure and Applied Mathematics}, 24:332--338,
  1890.

\bibitem[KKMS09]{KhandekarKMS09}
Rohit Khandekar, Tracy Kimbrel, Konstantin Makarychev, and Maxim Sviridenko.
\newblock {On Hardness of Pricing Items for Single-Minded Bidders}.
\newblock In {\em APPROX-RANDOM}, pages 202--216, 2009.

\bibitem[KMN10]{KarlinMN10}
Anna~R. Karlin, Claire Mathieu, and C.~Thach Nguyen.
\newblock {Integrality Gaps of Linear and Semi-definite Programming Relaxations
  for Knapsack}.
\newblock {\em CoRR}, abs/1007.1283, 2010.

\bibitem[KMR11]{KrauthgamerMR11}
Robert Krauthgamer, Aranyak Mehta, and Atri Rudra.
\newblock {Pricing Commodities}.
\newblock {\em Theor. Comput. Sci.}, 412(7):602--613, 2011.
\newblock Also in WAOA'07.

\bibitem[RS84]{RobertsonS84}
Neil Robertson and Paul~D. Seymour.
\newblock {Graph Minors. III. Planar Tree-Width}.
\newblock {\em J. Comb. Theory, Ser. B}, 36(1):49--64, 1984.

\bibitem[SA90]{SheraliA90}
Hanif~D. Sherali and Warren~P. Adams.
\newblock {A Hierarchy of Relaxations Between the Continuous and Convex Hull
  Representations for Zero-One Programming Problems}.
\newblock {\em SIAM J. Discrete Math.}, 3(3):411--430, 1990.

\bibitem[WJ04]{WainwrightJ04}
Martin~J Wainwright and Michael~I Jordan.
\newblock Treewidth-based conditions for exactness of the sherali-adams and
  lasserre relaxations.
\newblock {\em Univ. California, Berkeley, Technical Report}, 671, 2004.

\end{thebibliography}

\newpage
\appendix

\section*{Appendix}

\section{Generating (LP-$r$) from Lift-and-project operations}
\label{appendix: sherali-adams}

In this section we show that the constraints in (LP-$r$) can be automatically generated by applying $r$ rounds of Sherali-Adams lift-and-projects on the base LP. For an overview of Sherali-Adams hierarchy and the corresponding lift-and-project operations, we refer the reader to a survey by Chlamtac and Tulsiani \cite{hierarchies-survey}. There are many ways to write a natural LP relaxation for this problem, but one natural choice is the following: For each vertex $v$ and each possible price $i \in [P]$, we introduce variable $x(v,i)$, whose supposed value is to indicate that the price of $v$ is $i$. For each pair of vertices $u,v \in V$, we have variable $z(u,v,i,j)$ to be an indicator of $p(u) = i$ and $p(v) = j$.

\begin{eqnarray*}
\mbox{(LP')} & & \\
 & \max & \sum_{(u,v) \in E} \sum_{i+j \leq B_e} (i+j) z(u,v,i,j) \\
 & \mbox{s.t.} & \sum_{i \in [P]} x(v, i) = 1 \mbox{ \hspace{0.2in} for all $v \in V$} \\
 & & \sum_{i,j \in [P]} z(u,v,i,j) = 1 \mbox{ \hspace{0.2in} for all $u,v \in V$}\\
 & & z(u,v,i,j) \geq x(u,i) + x(v,j) -1 \mbox{\hspace{0.2in} for all $u, v \in V$, $i,j \in [P]$} \\
 & & x(u,i), z(u,v,i,j) \in [0,1]
\end{eqnarray*}

%\noindent\begin{align*}
%\label{eq:LP}\tag{\footnotesize LP'}
%\footnotesize
%\hspace{-5pt}
%\noindent\begin{aligned}
%%\mbox{(LP-$r$)}\\
%\lefteqn{\max \quad \sum_{(u,v) \in E} \sum_{i+j \leq B_e} (i+j) z(u,v,i,j)\hphantom{\max}  \\
%{\rm s.t.} && \sum_{\beta \in [P]^{T}}  y(S \cup T, \alpha \cup \beta)  &=  y(S, \alpha)  &&   \forall S,T \subseteq V: |S \cup T| \leq r, S \cap T = \emptyset\\
%&& 0 \leq y(S, \alpha)  &\leq 1 &&\forall S: \abs{S} \leq r, \forall \alpha \in [P]^S\\
%&& y(\emptyset, \emptyset) &= 1
%\end{aligned}
%\end{align*}

The last set of constraints enforces (in the integral world) that if $x(u,i) = 1$ and $x(v,j) =1$, then $z(u,v,i,j)$ must be set to $1$. It is easy to see that this LP has a bad integrality gap \parinya{Check this!}. We remark that the LP relaxation used by Khandekar et al. \cite{KhandekarKMS09} is equivalent to (LP') after two rounds of Sherali-Adams lift-and-project operations, and has an integrality gap upper bound of $4$ in general graphs.

Let $K_r$ be the polytope of feasible solution to (LP-$r$) and $K'$ be the feasible polytope for (LP'). Denote by ${\sf SA}^r(K')$ the polytope obtained after applying $r$ rounds of lift-and-projects to $K'$. We will not derive all the constraints that define ${\sf SA}^r(K')$, but instead we will only show that ${\sf SA}^{r}(K') \subseteq K_r$. This is sufficient for us to conclude that $O(k)$ rounds of Sherali-Adams hierarchy are enough to solve the graph pricing problem where the graph has tree-width at most $k$.

\begin{lemma}
For all $r \geq 2$, ${\sf SA}^{r}(K') \subseteq K_r$
\end{lemma}

\begin{proof}
We prove by induction on $r$ that constraints of (LP-$r$) can be generated by $r$ rounds of lift operations. For $r=2$, we show how to derive all the constraints of (LP-$2$) using at most $2$ rounds. We introduce variables $y'(S, \alpha)$ for each set $S \subseteq V, |S| \leq 2$ and for each assignment $\alpha \in [P]^S$, such that $y'(\set{u}, i) = x(u,i)$ for all $u \in V$ and $i \in [P]$. We let $y'(\emptyset, \emptyset) =1$. It is easy to see that variables $y'(S,\alpha)$ correspond exactly to the variables $y(S,\alpha)$ in (LP-2), and all constraints can be generated. The only nontrivial part is to show that, after two rounds of lifts, the constraint $y'(\set{u,v}, \set{i,j}) = z(u,v,i,j)$ holds for all $u,v \in V$ and $i,j \in [P]$. We sketch a proof of this fact here: First, apply $y'(u,i)$ to the third set of constraints to get
\[y'(u,i) \star z(u,v,i,j) \geq y'(\set{u,v}, \set{i,j})\]
By applying $z(u,v,i,j)$ to the inequality $y'(u,i) \leq y(\emptyset, \emptyset)$, we can get $y'(u,i) \star z(u,v,i,j) \leq z(u,v,i,j)$, so this implies $y'(\set{u,v}, \set{i,j}) \leq z(u,v,i,j)$. By summing over $i,j$, we get $1 = \sum_{i,j \in [P]} y'(\set{u,v}, \set{i,j}) \leq \sum_{i,j \in [P]} z(u,v,i,j) = 1$, so the inequality has to be tight for all $i,j$, i.e. $z(u,v,i,j) = y'(\set{u,v}, \set{i,j})$, as desired.

Next we show that, by applying one round of lift-and-project to (LP-$r$), we get all the constraints in (LP-$(r+1)$). It suffices to consider only the constraints that involve $S,T \subseteq V$ such that $|S \cup T| = r+1$. Let $S$ be a non-empty set (if it was empty, the claim follows trivially), so we can write $S = \set{v} \cup S'$ and $\alpha' = \alpha |_{S'}$. Notice that constraint $\sum_{\beta \in [P]^T} y(S' \cup T, \alpha' \cup \beta) = y(S', \alpha')$ belongs to (LP-$r$). By applying Sherali-Adams operation $\set{y(\set{v}, \alpha(v))} \star$ to both sides, we get the desired inequality.
\end{proof}

\begin{corollary}
If the input graph has tree-width at most $k$, then the Sherali-Adams polytope ${\sf SA}^{k+1}(K')$ has an integrality gap exactly one.
\end{corollary}

\end{document}